\documentclass[11pt]{article}

\usepackage[margin=1in]{geometry}
\usepackage{bbm}
\usepackage{bm}
\usepackage{cancel}
\usepackage{xcolor}
\usepackage{color}
\definecolor{gray}{gray}{0.5}
\definecolor{lightred}{rgb}{1,0.6,0.6}
\definecolor{darkgreen}{rgb}{0,0.5,0}
\definecolor{myorange}{rgb}{0.8,0.7,0.5}
\definecolor{darkblue}{rgb}{0.0,0.0,0.5}
\usepackage[colorlinks=true,breaklinks=true,bookmarks=true,urlcolor=darkblue,
citecolor=darkblue,linkcolor=darkblue,bookmarksopen=false,draft=false]{hyperref}
\usepackage{mathtools}
\usepackage{amsmath,amssymb,amsthm}
\usepackage[round]{natbib}

\DeclareMathOperator{\sgn}{sgn}
\newcommand{\reals}{\mathbb{R}}
\newcommand{\E}{\mathop{\mathbb{E}}}

\newcommand{\ones}{\mathbf{1}}

\newcommand{\C}{\mathcal{C}}

\newcommand{\DMI}{\textrm{DMI}}

\newcommand{\sigmatrue}{\sigma_{\text{true}}}
\newcommand{\restr}[2]{#1(#2)}
\newcommand{\relab}[3]{\pi_{#1}^{#2 \to #3}}

\usepackage[colorinlistoftodos,textsize=tiny]{todonotes}
\tikzset{notestyleraw/.append style={inner sep = 2pt}}
\newcommand{\Comments}{1}
\newcommand{\mynote}[2]{\ifnum\Comments=1\textcolor{#1}{#2}\fi}
\newcommand{\mytodo}[2]{\ifnum\Comments=1%
  \todo[linecolor=#1!80!black,backgroundcolor=#1,bordercolor=#1!80!black]{#2}\fi}

\ifnum\Comments=1
\fi

\newtheorem{theorem}{Theorem}
\newtheorem{lemma}[theorem]{Lemma}
\newtheorem{proposition}[theorem]{Proposition}
\newtheorem{corollary}[theorem]{Corollary}
\newtheorem{definition}[theorem]{Definition}
\newtheorem{assumption}[theorem]{Assumption}
\theoremstyle{definition}
\newtheorem{remark}[theorem]{Remark}
\newtheorem{example}[theorem]{Example}

\newcommand{\qeddiamond}{\leavevmode\unskip\penalty9999\hbox{}\nobreak\hfill\quad$\diamond$}
\AtEndEnvironment{remark}{\qeddiamond}
\AtEndEnvironment{example}{\qeddiamond}

\title{Joint-task truthfulness of the DMI mechanism}
\author{Rafael Frongillo\\
  Department of Computer Science\\
  University of Colorado Boulder\\
  \texttt{raf@colorado.edu}}
\date{\today}

\begin{document}
\maketitle

\begin{abstract}
  The Determinant Mutual Information (DMI) mechanism of \citet{kong2020dominantly,kong2024dominantly} is dominantly truthful within the class of \emph{consistent} reporting strategies, those that apply the same single-task strategy to every task.
  In settings where agents see multiple tasks before reporting, such as peer grading or peer review, it is natural to consider \emph{joint-task} strategies that may condition reports on the full signal vector.
  Perhaps surprisingly, we show that the DMI mechanism preserves truthful reporting as a best response among all joint-task strategies when other agents play consistent strategies, so that truthfulness remains a Bayes--Nash equilibrium in the joint-task class.
  Without the restriction of peers to consistent strategies, however, both dominant truthfulness and informed truthfulness fail against joint-task peer strategies.
\end{abstract}

\section{Introduction} \label{sec:intro}

Multi-task peer prediction mechanisms elicit multiple-choice answers from a population of agents without access to ground truth, by paying each agent based on how their reports correlate with peers' across a batch of similar tasks.
The Determinant Mutual Information (DMI) mechanism of \citet{kong2020dominantly,kong2024dominantly}, which pays each agent via determinants computed from empirical joint-report matrices, satisfies many desirable properties, including detail-free, informed-truthful, and dominantly truthful.
DMI is \emph{detail-free} in the sense that no prior knowledge of the joint signal distribution is required; moreover, it works with as few as $T \geq 2C$ tasks when each question has up to $C$ possible answers.
It is \emph{informed-truthful} in the sense that truthful reporting forms a Bayes--Nash equilibrium, and the truthful profile pays each agent strictly more than any uninformative profile in which reports are independent of signals.
It is \emph{dominantly truthful} in the sense that truthful reporting maximizes each agent's expected payment regardless of the strategies the other agents play~\citep[Theorem~5.3]{kong2020dominantly}.
Both truthfulness guarantees are stated with respect to the class of \emph{consistent} reporting strategies, in which each agent commits in advance to a single per-task kernel $\hat\sigma : \C \to \Delta(\C)$ and applies it independently to every task.

In many settings, agents may see some or all tasks in a batch before submitting their reports.
For example, a student may be sent all of their peer grading assignments at once, a crowdworker may have access to the entire batch of data to label, and a conference peer reviewer may be able to access all of their assigned papers at once.
In such cases, agents can choose their reports more holistically, and condition their joint report across all tasks on the joint vector of inputs or signals.
Such behavior might even be natural, as agents try to self-calibrate according to their prior, e.g.\ by making sure they do not fail all of their peers or reject every paper in their batch.
To capture this behavior, it is natural to consider \emph{joint-task strategies} $\sigma : \C^T \to \Delta(\C^T)$ that can choose reports based on all $T$ signals simultaneously.
As joint-task strategies are a strict superset of consistent strategies, an important question is thus whether DMI retains any truthfulness guarantees in this broader class.

We resolve this question, showing that truthful reporting remains a Bayes--Nash equilibrium of the DMI mechanism in the joint-task class, but that informed truthfulness and dominant truthfulness both fail.
Truthfulness follows from a stronger statement: when every peer plays a consistent strategy, truthful reporting is a best response among all joint-task strategies.
Without restricting peers to consistent strategies, however, dominance fails, even in the binary case (Example~\ref{ex:joint-task-peer-failure}).
The same example shows the failure of informed truthfulness as well (Remark~\ref{rem:not-informed-truthful}).

To prove truthfulness, we begin in \S~\ref{sec:binary} with the binary case $C = 2$, $T = 4$ (Theorem~\ref{thm:dmi-joint-strategy-binary}).
Here the payment is nonzero only when both signals and reports are distinct within each block, so the strategy boils down to choosing how to permute signals to reports, and the optimal strategy matches parity across blocks.
In \S~\ref{sec:general-C} we treat $C \geq 2$ and arbitrary task partitions $|T_\ell| \geq C$ (Theorem~\ref{thm:dmi-joint-strategy-general}) via a permutation expansion of $\det M_\ell$ over $C$-subsets of each block, with the same intuition; the optimum is in fact achieved by any matching-parity pair of per-block permutations (Remark~\ref{rem:other-optima}).

Our study is motivated by the observation of \citet{frongillo2026peer}, that DMI is generally not truthful once signals are real-valued and reports are coarsened, because differences in the signal magnitudes near the discretization boundaries break the subtle ties needed in the arguments below.

\section{Setting} \label{sec:setting}

We follow the multi-task peer prediction model of \citet{kong2020dominantly,kong2024dominantly}, and largely adopt the notation of \citet{frongillo2026peer}.

\subsection{Tasks, signals, and strategies}

There are $n \geq 2$ agents and $T$ tasks.
Each task is a multiple-choice question with finite \emph{choice space} $\C := [C] = \{1, 2, \ldots, C\}$, where $C \geq 2$.
At each task $t \in [T]$, agent~$i$ privately receives a \emph{signal} $s_i^t \in \C$ and then submits a \emph{report} $r_i^t \in \C$.
We write $S_i^t$ and $R_i^t$ for the corresponding random variables.
We abbreviate the agent's full signal vector as $s_i := (s_i^1, \ldots, s_i^T) \in \C^T$, and similarly for reports, and analogously for the random vectors $S_i, R_i$.

\begin{assumption}[A priori similar tasks; \protect{\citealp[Assumption~3.1]{kong2020dominantly}}] \label{assume:iid-tasks}
There exists a prior $U \in \Delta(\C^n)$ such that the signal vectors $(S_1^t, \ldots, S_n^t)_{t \in [T]}$ are i.i.d.\ copies of $(S_1, \ldots, S_n) \sim U$.
\end{assumption}

In \citet[Definition~3.3, Assumption~3.4]{kong2020dominantly}, agents must apply a single per-task strategy ``kernel'' consistently across tasks.
We work with the more general class of \emph{joint-task strategies}, in which an agent maps the entire signal vector to a distribution over the entire report vector.

\begin{definition}[Joint-task strategy] \label{def:joint-strategy}
  A \emph{joint-task strategy} for agent~$i$ is a map $\sigma_i : \C^T \to \Delta(\C^T)$.
\end{definition}

\begin{definition}[Consistent strategy] \label{def:consistent}
  A joint-task strategy $\sigma : \C^T \to \Delta(\C^T)$ is \emph{consistent} if there exists a per-task kernel $\hat\sigma : \C \to \Delta(\C)$ such that
  \[
    \sigma(s)(r) \;=\; \prod_{t=1}^T \hat\sigma(s^t)(r^t)
    \qquad \text{for all } s, r \in \C^T~.
  \]
\end{definition}

\begin{definition}[Truthful strategy] \label{def:truthful}
  The \emph{truthful} strategy is $\sigmatrue : s \mapsto \delta_s$.
\end{definition}

The truthful strategy is consistent, with kernel $\hat\sigma(s) = \delta_s$.

\subsection{The DMI mechanism}

For two $\C$-valued random variables $X, Y$, write $U_{X,Y} \in \reals^{C \times C}$ for the matrix of their joint distribution: $(U_{X,Y})_{x,y} = \Pr[X = x, Y = y]$.

\begin{definition}[Determinant mutual information; \protect{\citealp[Definition~4.1]{kong2020dominantly}}] \label{def:dmi}
The \emph{determinant mutual information} between $X$ and $Y$ is
\[
  \DMI(X; Y) \;:=\; \lvert \det U_{X,Y} \rvert ~.
\]
\end{definition}

The mechanism below uses the squared quantity $\DMI(X;Y)^2 = (\det U_{X,Y})^2$, which admits a simple unbiased estimator from a constant number of samples.

Identify $\C$ with $[C]:=\{1, \ldots, C\}$, and for $c \in \C$ let $\ones_c \in \reals^C$ be the standard basis vector.
Fix a partition of the tasks into two disjoint sets $[T] = T_1 \sqcup T_2$ with $|T_1|, |T_2| \geq C$.

\begin{definition}[Answer matrix] \label{def:answer-matrix}
For distinct agents $i, j \in [n]$ and $\ell \in \{1, 2\}$, the \emph{answer matrix} is
\[
  M_\ell^{ij}
  \;:=\;
  \sum_{t \in T_\ell} \ones_{r_i^t}\, \ones_{r_j^t}^{\top}
  \;\in\; \reals^{C \times C}~,
\]
whose $(c, c')$ entry counts the number of tasks $t \in T_\ell$ with $(r_i^t, r_j^t) = (c, c')$.
\end{definition}

\begin{definition}[DMI mechanism; \protect{\citealp[\S~5]{kong2020dominantly}}] \label{def:dmi-mech}
The \emph{DMI mechanism} pays each agent~$i$
\[
  M_\DMI^{\,i} \;:=\; \sum_{j \neq i} \det M_1^{ij} \cdot \det M_2^{ij}~.
\]
\end{definition}

When both $i$ and $j$ play consistent strategies, the expectation $\E[\det M_\ell^{ij}]$ is proportional to $\det U_{R_i, R_j}$~\citep[Claim~5.2]{kong2020dominantly}, so $\E[M_\DMI^{\,i}]$ is proportional to $\sum_{j \neq i} \DMI(R_i; R_j)^2$.
Within the class of consistent strategies, the mechanism is dominantly truthful~\citep[Theorem~5.3]{kong2020dominantly}.
It is left open by these results whether truthfulness remains a best response when an agent may deviate to an arbitrary joint-task strategy.

\section{Binary reports} \label{sec:binary}

We begin with the more concrete setting of binary reports ($C=2$) and four total tasks ($T=4$), the minimal setting to study the DMI mechanism.
Throughout we will write $\C = \{H, L\}$ and consider the partition $T_1 = \{1, 2\}, T_2 = \{3, 4\}$.
Our main results for this binary setting are as follows.

\begin{theorem} \label{thm:dmi-joint-strategy-binary}
  Let $C = 2$ and $T = 4$.
  Suppose every agent $j \neq i$ plays a consistent strategy.
  Then the truthful strategy $\sigmatrue$ maximizes $\E[M_\DMI^{\,i}]$ among all joint-task strategies $\sigma_i : \C^T \to \Delta(\C^T)$.
\end{theorem}

\begin{corollary}
  When $C = 2$ and $T = 4$, the truthful strategy $\sigmatrue$ is a Bayes--Nash equilibrium of the DMI mechanism among joint-task strategies.
\end{corollary}

We first prove the theorem for $n = 2$ (one agent, one peer), and then extend the conclusion to $n \geq 2$.
We show in \S~\ref{sec:peer-consistency-necessary} that dominant truthfulness fails against general joint-task peers.

\subsection{Proof of Theorem~\ref{thm:dmi-joint-strategy-binary}}

We first consider $n = 2$, from the point of view of agent~$i$ with peer agent~$j$.
The peer plays a consistent strategy with kernel $\hat\sigma_j$, so $R_j^t \sim \hat\sigma_j(S_j^t)$ for each task $t$, independently across tasks.
Define
\begin{align*}
  p(s)
  \;&:=\; \Pr[R_j^t = H \mid S_i^t = s]\\
  \;&=\; \sum_{s' \in \C} \Pr[S_j^t = s' \mid S_i^t = s] \, \hat\sigma_j(s')(H)~,
\end{align*}
which by Assumption~\ref{assume:iid-tasks} and consistency of $\hat\sigma_j$ does not depend on $t$.
The conditional law of $S_j^t$ given $S_i^t$ comes from the prior $U$.

The following lemma is the key building block for the $C=2$ case.
It says agent $i$ will have a zero expected value for the determinant for any block of two reports whenever (a) their reports are the same, or (b) their \emph{signals} are the same.
Case (a) follows directly from the definition of the determinant, but (b) is more subtle, and has to do with the exact cancellation of terms when agent $j$ plays a consistent strategy.
Recalling that the final payment is a product of independent determinant terms, the implication is that an agent can only have a nonzero expected payoff when their signals are distinct within each block, and so are their reports.
From there, the strategy boils down to a simple decision for each block: whether to match reports to signals or reverse them.
Proposition~\ref{prop:two-agent-pointwise} merely observes that matching in both blocks (or reversing in both) is optimal.

\begin{lemma} \label{lem:dmi-det-formula}
  Fix $\ell \in \{1, 2\}$, let $\{t, t'\} = T_\ell$, and for ease of exposition let us write $r := r_i^t$, $r' := r_i^{t'}$, $S := S_i^t$, $S' := S_i^{t'}$.
  For any reports $(r, r') \in \{H, L\}^2$,
  \[
    \E[\det M_\ell^{ij} \mid S, S']
    \;=\;
    \begin{cases}
      0,
        & r = r' \text{ or } S = S',\\
      p(H) - p(L),
        & (r, r') = (S, S'),\\
      p(L) - p(H),
        & (r, r') = (S', S)~.
    \end{cases}
  \]
\end{lemma}

\begin{proof}
  Further write $R := R_j^t$, $R' := R_j^{t'}$ for the peer's reports.
  If $r = r'$ then $M_\ell^{ij}$ has rank one, so $\det M_\ell^{ij} = 0$ deterministically.
  Otherwise, by symmetry it suffices to take $(r, r') = (H, L)$, in which case
  \[
    \det M_\ell^{ij} \;=\;
    \begin{cases}
      +1, & (R, R') = (H, L),\\
      -1, & (R, R') = (L, H),\\
      0,  & \text{otherwise}~.
    \end{cases}
  \]
  Across tasks the signal vectors are i.i.d.\ (Assumption~\ref{assume:iid-tasks}), so conditional on $(S, S')$ the pair $(S_j^t, R)$ is independent of $(S_j^{t'}, R')$.
  Hence
  \begin{align*}
    \E[\det M_\ell^{ij} \mid S, S']
    &= p(S)\bigl(1 - p(S')\bigr) - \bigl(1 - p(S)\bigr)\,p(S')\\
    &= p(S) - p(S')~.
  \end{align*}
  The case $(r, r') = (L, H)$ gives $p(S') - p(S)$ by the same argument.

  If $S = S'$ the right-hand side is zero.
  Otherwise $\{S, S'\} = \{H, L\}$, and the right-hand side is $p(H) - p(L)$ when reports match signals and $p(L) - p(H)$ when reports are swapped.
\end{proof}

\begin{proposition}[Two-agent best response] \label{prop:two-agent-pointwise}
  Suppose agent $j$ plays a consistent strategy.
  For any signal realization $s_i \in \C^T$, $\sigmatrue$ maximizes
  \[
    \E\bigl[\det M_1^{ij} \cdot \det M_2^{ij} \,\bigm|\, S_i = s_i\bigr]
  \]
  over all joint-task strategies $\sigma_i : \C^T \to \Delta(\C^T)$.
\end{proposition}

\begin{proof}
Fix any reports $r_i \in \C^T$.
Conditional on $S_i$, the determinants $\det M_1^{ij}$ and $\det M_2^{ij}$ depend on disjoint sets of peer signals and reports, so they are conditionally independent (Assumption~\ref{assume:iid-tasks}, consistency of $\hat\sigma_j$).
Hence
\[
  \E\bigl[\det M_1^{ij} \cdot \det M_2^{ij} \mid S_i\bigr]
  \;=\;
  \E[\det M_1^{ij} \mid S_i^1, S_i^2] \cdot \E[\det M_2^{ij} \mid S_i^3, S_i^4]~.
\]

If $s_i^1 = s_i^2$ or $s_i^3 = s_i^4$, Lemma~\ref{lem:dmi-det-formula} forces one factor to be zero for every report, so the conditional expected payment is zero, and $\sigmatrue$ trivially achieves the maximum.

Otherwise $s_i^1 \neq s_i^2$ and $s_i^3 \neq s_i^4$.
If $r_i^1 = r_i^2$ or $r_i^3 = r_i^4$, again one factor vanishes.
In all remaining cases each factor equals $\pm(p(H) - p(L))$.
The product of these factors equals $(p(H) - p(L))^2$ exactly when (a) the reports match the signals on both pairs (truthful) or (b) reverse them on both pairs; the product is $-(p(H) - p(L))^2$ otherwise.
In particular, the maximum $(p(H) - p(L))^2$ is attained by $\sigmatrue$.
\end{proof}

Finally, Theorem~\ref{thm:dmi-joint-strategy-binary} follows by observing that the payoff is simply a sum of two-agent payoffs.

\begin{proof}[Proof of Theorem~\ref{thm:dmi-joint-strategy-binary}]
Fix agent~$i$ and any signal realization $S_i = s_i$.
The payment decomposes as
\[
  M_\DMI^{\,i} \;=\; \sum_{j \neq i} \det M_1^{ij} \cdot \det M_2^{ij}~,
\]
and each summand depends only on the reports of $i$ and the single peer $j$.
By hypothesis each peer $j \neq i$ plays a consistent strategy, so Proposition~\ref{prop:two-agent-pointwise} applied to the pair $(i, j)$ shows that $\sigmatrue$ maximizes $\E[\det M_1^{ij} \cdot \det M_2^{ij} \mid S_i = s_i]$.
Summing over $j$, $\sigmatrue$ pointwise maximizes
\[
  \E[M_\DMI^{\,i} \mid S_i = s_i]~.
\]
By the tower property, $\sigmatrue$ also maximizes the unconditional expectation $\E[M_\DMI^{\,i}]$.
\end{proof}

\subsection{Necessity of peer consistency} \label{sec:peer-consistency-necessary}

The hypothesis in Theorem~\ref{thm:dmi-joint-strategy-binary} that peers be consistent cannot be dropped: even in the binary case, dominance against joint-task peers fails trivially.

\begin{example} \label{ex:joint-task-peer-failure}
Take $n = 2$, $C = 2$, $T = 4$, with prior $U$ uniform on $\{(H, H), (L, L)\} \subset \{H, L\}^2$, so $S_i^t = S_j^t$ deterministically with $S_i^t \sim \mathrm{Unif}\{H, L\}$.
Let agent~$j$ play the joint-task strategy $\sigma_j : s_j \mapsto \delta_{(H, L, H, L)}$ that ignores their signals and always reports $(H, L, H, L)$.
This $\sigma_j$ is deterministic and joint-task but not consistent: applied to the constant input $s_j = (H, H, H, H)$, it would require $\hat\sigma(H) = H$ (from $r_j^1 = H$) and simultaneously $\hat\sigma(H) = L$ (from $r_j^2 = L$).
Under truthful reporting, agent~$i$'s reports are $r_i^t = S_i^t$, i.i.d.\ uniform across tasks.
Conditioning on $S_i^1, S_i^2$,
\[
  \E[\det M_1^{ij} \mid S_i^1, S_i^2] \;=\; \ones[(S_i^1, S_i^2) = (H, L)] \;-\; \ones[(S_i^1, S_i^2) = (L, H)]~,
\]
giving $\E[\det M_1^{ij}] = 0$.
By independence of the tasks in $T_1$ from those in $T_2$, $\E[\det M_1^{ij} \cdot \det M_2^{ij}] = 0$.
Agent~$i$ may instead deviate to $\sigma_i = \sigma_j$ and play $r_i = (H, L, H, L)$, yielding $M_1^{ij} = M_2^{ij} = I$ and payment $1$ deterministically.
The joint-task deviation strictly dominates truthful, so dominant truthfulness fails against this joint-task peer.
The construction extends to arbitrary $C$ and $T \geq 2C$: take any deterministic peer report $r_j \in \C^T$ whose values cover $\C$ within each block, e.g., $r_j = (1, 2, \ldots, C, 1, 2, \ldots, C)$ for $T=2C$.
If agent $i$ plays truthfully, their payoff is again $\E[M_\DMI^{\,i}] = 0$, while choosing $r_i = r_j$ deterministically nets $\prod_\ell \prod_c |\{t \in T_\ell : r_j^t = c\}| \geq 1$.
\end{example}

\begin{remark}[Informed-truthfulness also fails] \label{rem:not-informed-truthful}
In the $C = 2$, $T = 4$ setting of Example~\ref{ex:joint-task-peer-failure}, the same construction shows that DMI is not \emph{informed-truthful} in the joint-task class: there is an uninformative profile that strictly outpays the truthful profile.
Under the truthful profile (both agents play $\sigmatrue$), $r_i^t = r_j^t = S_i^t$ is i.i.d.\ uniform across tasks, so for each block $\det M_\ell^{ij} \in \{0, 1\}$ takes value $1$ iff $r_i^1 \neq r_i^2$, which happens with probability $\tfrac12$.
The two blocks are independent, so $\E[M_\DMI^{\,i}] = \tfrac14$.
Under the uninformative profile in which both agents play the constant strategy $\sigma : s \mapsto \delta_{(H, L, H, L)}$, we have $M_1^{ij} = M_2^{ij} = I$ and $M_\DMI^{\,i} = 1$ deterministically.
Thus the uninformative profile pays $1 > \tfrac14$, violating informed truthfulness.
\end{remark}

\section{General setting} \label{sec:general-C}

We now lift the binary-report argument of \S~\ref{sec:binary} from $C = 2$ to arbitrary $C \geq 2$ and arbitrary partitions of the tasks into blocks $|T_1|, |T_2| \geq C$.
When $C = |T_1| = |T_2|$, the nonzero-payoff observation generalizes nicely: agent $i$'s expected payoff is nonzero only when both their reports and their \emph{signals} are all distinct on each block.
Whereas the $C = 2$, $T = 4$ case offered a binary choice (match or reverse) on each block, agent $i$ now effectively chooses a permutation to map signals to reports on each block.
The conditional expected payoff is proportional to the product of the two permutations' signs, so the optimum is achieved whenever the signs agree (Remark~\ref{rem:other-optima}); crucially, truthfulness is one such strategy.
The case of more than $2C$ tasks generalizes this observation further, where now we make use of the \emph{permutation expansion} of $\det M_\ell^{ij}$, but the underlying logic remains.

\begin{theorem} \label{thm:dmi-joint-strategy-general}
  Let $C \geq 2$, any $T \geq 2C$, and consider any partition $[T] = T_1 \sqcup T_2$ with $|T_1|, |T_2| \geq C$.
  Suppose every agent $j \neq i$ plays a consistent strategy.
  Then the truthful strategy $\sigmatrue$ maximizes $\E[M_\DMI^{\,i}]$ among all joint-task strategies $\sigma_i : \C^T \to \Delta(\C^T)$.
\end{theorem}

\subsection{Permutation expansion} \label{sec:perm-expansion}

\begin{definition}[Conditional report matrix] \label{def:Q-matrix}
  The \emph{conditional report matrix} $Q \in \reals^{C \times C}$ has entries
  \begin{align*}
    Q_{c, c'} \;&:=\; \Pr[R_j^t = c' \mid S_i^t = c]\\
    \;&=\; \sum_{c'' \in \C} \Pr[S_j^t = c'' \mid S_i^t = c]\,\hat\sigma_j(c'')(c')~,
  \end{align*}
  which by Assumption~\ref{assume:iid-tasks} and consistency of $\hat\sigma_j$ does not depend on $t$.
\end{definition}

When $C = 2$, $Q$ collapses to the binary kernel: $Q_{c, H} = p(c)$ and $Q_{c, L} = 1 - p(c)$, giving $\det Q = p(H) - p(L)$.

For a subset $J \subseteq [T]$ and tuple $x \in \C^T$, write $\restr{x}{J} := \{x^t : t \in J\}$ for the set of values $x$ takes on $J$.
When $|J| = |\C| = C$, the equality $\restr{x}{J} = \C$ holds exactly when these values are all distinct.
When $\restr{s_i}{J} = \restr{r_i}{J} = \C$, the assignment $s_i^t \mapsto r_i^t$ for $t \in J$ uniquely determines a permutation of $\C$, which we denote $\relab{J}{s_i}{r_i} \in \mathfrak{S}_C$.

\begin{lemma}[Permutation expansion] \label{lem:perm-expansion}
  For each $\ell \in \{1, 2\}$ and any reports $r_i \in \C^T$,
  \[
    \E[\det M_\ell^{ij} \mid S_i^t : t \in T_\ell]
    \;=\;
    \det Q \;\cdot\;
    \sum_{\substack{J \subseteq T_\ell\\ |J| = C}}
    \ones\!\bigl[\restr{S_i}{J} = \restr{r_i}{J} = \C\bigr] \, \sgn(\relab{J}{S_i}{r_i})~.
  \]
\end{lemma}

\begin{proof}
  Define matrices $A, B \in \{0, 1\}^{\C \times T_\ell}$ entry-wise by
  \begin{equation} \label{eq:AB-def}
    A_{c, t} \;:=\; \ones[r_i^t = c]~, \qquad B_{c, t} \;:=\; \ones[R_j^t = c]~,
  \end{equation}
  for $c \in \C$ and $t \in T_\ell$.
  Then $M_\ell^{ij} = A B^\top$: we have $(A B^\top)_{c, c'} = \sum_{t \in T_\ell} \ones[r_i^t = c]\,\ones[R_j^t = c']$, which by Definition~\ref{def:answer-matrix} equals $(M_\ell^{ij})_{c, c'}$.
  The Cauchy--Binet formula \citep[see, e.g.,][Corollary~(Cauchy-Binet) in \S~4.6, p.~214]{broida2012linear} expands the determinant as
  \begin{equation} \label{eq:cb-expansion}
    \det M_\ell^{ij} \;=\; \sum_{\substack{J \subseteq T_\ell\\ |J| = C}} \det A_J \cdot \det B_J~,
  \end{equation}
  where $A_J, B_J \in \{0, 1\}^{\C \times J}$ are the column-restrictions of $A, B$ to tasks in $J$, regarded as $C \times C$ matrices.
  (The product $\det A_J \cdot \det B_J$ is independent of the column ordering.)

  Let us now consider the expected value.
  As $r_i$ is fixed, from eq.~\eqref{eq:AB-def} we can see that $A$ is deterministic, while $B$ is random through the dependence on $R_j$; the same applies to $A_J$ and $B_J$.
  Thus we have
  \begin{equation} \label{eq:cb-expectation}
    \E\bigl[\det M_\ell^{ij} \,\big|\, S_i^t : t \in T_\ell\bigr]
    =
    \sum_{\substack{J \subseteq T_\ell\\ |J| = C}} \det A_J \cdot \E\bigl[\det B_J \,\big|\, S_i^t : t \in J\bigr]~.
  \end{equation}
  Let us fix some $J \subseteq T_\ell, |J| = C$ in this sum.
  By Assumption~\ref{assume:iid-tasks} and consistency of $\hat\sigma_j$, the random variables $\{R_j^t : t \in J\}$ are mutually independent given $\{S_i^t : t \in J\}$, with $\Pr[R_j^t = c' \mid S_i^t = c] = Q_{c, c'}$.
  Applying the Leibniz formula \citep[see, e.g.,][\S~4.1]{broida2012linear} to $\det B_J$, we have
  \begin{equation} \label{eq:detB-tilde}
    \E\bigl[\det B_J \,\big|\, S_i^t : t \in J\bigr]
    \;=\; \sum_{\tau \in \mathfrak{S}_C} \sgn(\tau) \prod_{t \in J} Q_{S_i^t,\, \tau(t)}
    \;=\; \det \tilde B~,
  \end{equation}
  where $\tilde B \in \reals^{\C \times J}$ has entries $\tilde B_{c, t} := Q_{S_i^t, c}$.

  Just as in the proof of Lemma~\ref{lem:dmi-det-formula}, we can see from eq.~\eqref{eq:detB-tilde} that only $J$ terms with $\restr{r_i}{J} = \restr{S_i}{J} = \C$ contribute to the sum in eq.~\eqref{eq:cb-expectation}:
  when $\restr{r_i}{J} \neq \C$, two columns of $A_J$ coincide and $\det A_J = 0$;
  when $\restr{S_i}{J} \neq \C$, two columns of $\tilde B$ coincide and $\det \tilde B = 0$.
  For the remaining $J$ with $\restr{r_i}{J} = \restr{S_i}{J} = \C$, the relabeling $\relab{J}{r_i}{S_i}$ is defined.
  The product $A_J \tilde B^\top$ has entry
  \[
    (A_J \tilde B^\top)_{c, c'}
    \;=\; \sum_{t \in J} \ones[r_i^t = c]\, Q_{S_i^t, c'}
    \;=\; Q_{\relab{J}{r_i}{S_i}(c),\, c'}~,
  \]
  since the indicator picks out the unique $t \in J$ with $r_i^t = c$, whose signal $S_i^t$ is $\relab{J}{r_i}{S_i}(c)$ by definition.
  We now see that $A_J \tilde B^\top$ is just $Q$ with rows permuted by $\relab{J}{r_i}{S_i}$, giving
  \[
    \det A_J \cdot \det \tilde B
    \;=\; \det(A_J \tilde B^\top)
    \;=\; \sgn(\relab{J}{S_i}{r_i})\, \det Q~,
  \]
  where the last equality uses $\sgn(\relab{J}{r_i}{S_i}) = \sgn(\relab{J}{S_i}{r_i})$ as inverses.
  In light of eq.~\eqref{eq:cb-expansion} and eq.~\eqref{eq:cb-expectation}, summing over $J$ completes the proof.
\end{proof}

For $C = 2$ with $|T_\ell| = 2$, the only $C$-subset of $T_\ell$ is $T_\ell$ itself, so each $A_\ell$ reduces to a single summand: zero when $\restr{s_i}{T_\ell} \neq \C$ or $\restr{r_i}{T_\ell} \neq \C$, and $\sgn(\relab{T_\ell}{s_i}{r_i}) \in \{-1, 1\}$ otherwise.
Combined with $\det Q = p(H) - p(L)$, Lemma~\ref{lem:perm-expansion} recovers Lemma~\ref{lem:dmi-det-formula}.

\subsection{Proof of Theorem~\ref{thm:dmi-joint-strategy-general}} \label{sec:proof-general}

For a signal realization $s_i \in \C^T$, reports $r_i \in \C^T$, and block $\ell \in \{1, 2\}$, write
\begin{align*}
  N_\ell(s_i) \;&:=\; \bigl|\{ J \subseteq T_\ell : |J| = C,\; \restr{s_i}{J} = \C \}\bigr|~,\\
  A_\ell(s_i, r_i) \;&:=\; \sum_{\substack{J \subseteq T_\ell\\ |J| = C}} \ones[\restr{s_i}{J} = \restr{r_i}{J} = \C] \, \sgn(\relab{J}{s_i}{r_i})~,
\end{align*}
the count of $C$-subsets of $T_\ell$ on which $s_i$ is a bijection, and the signed sum from Lemma~\ref{lem:perm-expansion}, so that
\[
  \E[\det M_\ell^{ij} \mid S_i^t : t \in T_\ell] \;=\; \det Q \cdot A_\ell(S_i, r_i)~.
\]
We follow the proof strategy of the binary case in \S~\ref{sec:binary}.

To begin, the following is the key observation used in the optimality of $\sigmatrue$.
\begin{lemma} \label{lem:A-bound}
  For all $s_i, r_i \in \C^T$ and $\ell \in \{1, 2\}$, $|A_\ell(s_i, r_i)| \leq N_\ell(s_i)$, with equality when $r_i^t = s_i^t$ for all $t \in T_\ell$.
\end{lemma}

\begin{proof}
  Each summand in $A_\ell$ has absolute value at most $1$, and attains absolute value $1$ on the set
  $\{J : \restr{s_i}{J} = \restr{r_i}{J} = \C\} \subseteq \{J : |J|=C, \restr{s_i}{J} = \C\}$.
  Thus $|A_\ell| \leq N_\ell$.
  When $r_i^t = s_i^t$ for every $t \in T_\ell$, $\relab{J}{s_i}{r_i}$ is the identity permutation for every $J \subseteq T_\ell$ with $\restr{s_i}{J} = \C$, giving $\sgn(\relab{J}{s_i}{r_i}) = 1$ on each.
  We conclude $A_\ell(s_i, s_i) = N_\ell(s_i)$.
\end{proof}

\begin{proposition}[Two-agent best response] \label{prop:two-agent-general}
  Suppose agent $j$ plays a consistent strategy.
  For any signal realization $s_i \in \C^T$, $\sigmatrue$ maximizes
  \[
    \E\bigl[\det M_1^{ij} \cdot \det M_2^{ij} \,\bigm|\, S_i = s_i\bigr]
  \]
  over all joint-task strategies $\sigma_i : \C^T \to \Delta(\C^T)$.
\end{proposition}

\begin{proof}
  Without loss of generality consider a deterministic strategy $\sigma_i : s_i \mapsto r_i$; the randomized case follows by linearity of expectation.
  By Assumption~\ref{assume:iid-tasks} and consistency of $\hat\sigma_j$, the peer signals and reports on $T_1$ are independent of those on $T_2$ given $S_i$, so
  \begin{align*}
    \E\bigl[\det M_1^{ij} \cdot \det M_2^{ij} \,\bigm|\, S_i\bigr]
    &= \E[\det M_1^{ij} \mid S_i^t : t \in T_1] \cdot \E[\det M_2^{ij} \mid S_i^t : t \in T_2]\\
    &= (\det Q)^2 \cdot A_1(s_i, r_i) \cdot A_2(s_i, r_i)~.
  \end{align*}
  Applying Lemma~\ref{lem:A-bound} to both factors,
  \begin{align*}
    A_1(s_i, r_i) \cdot A_2(s_i, r_i)
    &\leq |A_1| \cdot |A_2|\\
    &\leq N_1(s_i) \cdot N_2(s_i)\\
    &= A_1(s_i, s_i) \cdot A_2(s_i, s_i)~.
  \end{align*}
  Since $(\det Q)^2 \geq 0$, $\sigmatrue$ maximizes the conditional expectation.
\end{proof}

\begin{proof}[Proof of Theorem~\ref{thm:dmi-joint-strategy-general}]
  We follow the proof of Theorem~\ref{thm:dmi-joint-strategy-binary} almost identically.
  As each peer $j \neq i$ plays a consistent strategy, Proposition~\ref{prop:two-agent-general} for the pair $(i, j)$ shows that $\sigmatrue$ maximizes $\E[\det M_1^{ij} \cdot \det M_2^{ij} \mid S_i = s_i]$.
  Summing over $j$, then, $\sigmatrue$ maximizes $\E[M_\DMI^{\,i} \mid S_i = s_i]$, and by the tower property also the unconditional expectation $\E[M_\DMI^{\,i}]$.
\end{proof}

\begin{remark}[Other optimal strategies] \label{rem:other-optima}
  Recall that in the binary case, there were two optimal strategies: being truthful in both blocks, or the strategy that swaps $H\mapsto L$ and $L\mapsto H$ in both.
  In the general setting we see a similar phenomenon, but instead with the parity of the induced permutation in place of the truthful/swap decision.
  Specifically, the optimal set is all strategies of the form $r_i^t = \sigma_\ell(s_i^t)$ for each $t \in T_\ell$, $\ell \in \{1, 2\}$, where $\sigma_1, \sigma_2 \in \mathfrak{S}_C$ are any permutations satisfying $\sgn(\sigma_1) = \sgn(\sigma_2)$.
  For $C \geq 3$, matching parity is strictly weaker than picking the same permutation in both blocks, so the optimal set is substantially larger than in the binary case, with truthful being one of many optima.
\end{remark}

\subsection*{Acknowledgements}

The author thanks Ian Kash and Mary Monroe for helpful discussions.

\bibliographystyle{plainnat}

\end{document}